\documentclass[american]{article}
\usepackage[T1]{fontenc}
\usepackage[utf8]{inputenc}
\usepackage{geometry}
\usepackage{graphicx}
\usepackage{color}
\usepackage{babel}
\usepackage{amsmath}
\usepackage{amsthm}
 \usepackage{amssymb}
\usepackage[pdfusetitle,
 bookmarks=true,bookmarksnumbered=false,bookmarksopen=false,
 breaklinks=false,pdfborder={0 0 1},backref=false,colorlinks=true]
 {hyperref}

\makeatletter

\newcommand{\lyxaddress}[1]{
	\par {\raggedright #1
	\vspace{1.4em}
	\noindent\par}
}
\theoremstyle{plain}
\newtheorem{thm}{\protect\theoremname}[section]
\theoremstyle{plain}
\newtheorem{lem}[thm]{\protect\lemmaname}
\ifx\proof\undefined
\newenvironment{proof}[1][\protect\proofname]{\par
	\normalfont\topsep6\p@\@plus6\p@\relax
	\trivlist
	\itemindent\parindent
	\item[\hskip\labelsep\scshape #1]\ignorespaces
}{%
	\endtrivlist\@endpefalse
}
\providecommand{\proofname}{Proof}
\fi
\theoremstyle{remark}

\theoremstyle{remark}
\newtheorem{rem}[thm]{\protect\remarkname}
\theoremstyle{plain}
\newtheorem{prop}[thm]{\protect\propositionname}
\newtheorem{Definition}{Definition}
\@ifundefined{date}{}{\date{}}
\makeatother

\providecommand{\claimname}{Claim}
\providecommand{\lemmaname}{Lemma}
\providecommand{\propositionname}{Proposition}
\providecommand{\remarkname}{Remark}
\providecommand{\theoremname}{Theorem}

\begin{document}
\global\long\def\norm#1{\left\Vert #1\right\Vert }%
\title{Two-Point Vortex Confinement in a simply connected domain}
\author{Slim Ibrahim$^{a}$, Ruixun Qin$^{b}$ and Shengyi Shen$^{b}$}
\maketitle

\lyxaddress{\begin{center}
$^{a}$ Department of Mathematics \& Statistics, University of Victoria\\
3800 Finnerty Rd, Victoria, BC V8P 5C2\\
ibrahims@uvic.ca\\
$^{b}$ School of Mathematics and Statistics, Ningbo University, Ningbo 315211,
P. R. China\\
qinruixunmath@outlook.com; shenshengyi@nbu.edu.cn
\par\end{center}}
\begin{abstract}
This paper investigates the vortex confinement property of the two-point vortex system in a planar domain. We compute the time over which initial point vortices around a stable stationary point remain within a slightly larger ball. In particular, we show that this concentration persists indefinitely regardless of the vorticity strengths. In the borderline of the stability condition, we show that this time becomes a power law, if in addition, one  relaxes the size of the stability ball. 
\end{abstract}

\section{Introduction}
The motion of an inviscid and incompressible fluid inside a simply connected domain $\Omega\subseteq\mathbb{R}^{2}$ is governed by the Euler equations 
\begin{eqnarray*}
\partial_t u + u \cdot \nabla u = -\nabla p\quad\mbox{in}\quad\Omega\times\mathbb{R}_+\\
u\cdot n=0\quad\mbox{on}\quad\partial\Omega\times\mathbb R_+\\ \text{div } u(t, \cdot) = 0\quad\mbox{in}\quad\Omega\times\mathbb R_+,
\end{eqnarray*}
where the unknowns are the velocity field of fluid particles ${\bf u}: \mathbb{R}_+ \times \Omega \to \mathbb{R}^2$, the fluid pressure $p : \mathbb{R}_+ \times \Omega \to \mathbb{R}$. The vector ${n}$ is the outward normal vector to $\partial\Omega$. Defining the vorticity as $\omega:=\partial_{1}u_{2}-\partial_{2}u_{1}$, the vorticity formulation of the initial value problem for the 2D Euler equations reads as 
\begin{equation}
\begin{cases}
\partial_{t}\omega\left(x,t\right)+u\left(x,t\right)\cdot\nabla\omega\left(x,t\right)=0, & \quad\mbox{in}\quad\Omega\times\mathbb{R}_+\\
\omega\left(x,0\right)=\omega_{0}\left(x\right), &\quad\mbox{in}\quad\Omega\\
\nabla\cdot u\left(x,t\right)=0, & \quad\mbox{in}\quad\Omega\times\mathbb{R}_+\\
u\left(x,t\right)\cdot{n}=0, & \quad\mbox{on}\quad\partial\Omega\times\mathbb R_+
\end{cases}\label{eq:sys_euler}
\end{equation}
The vorticity $\omega(t)$ can be seen as a rearrangement of the initial vorticity $\omega_0$ through the flow map defined by the velocity field $u$.\\
Recall that the velocity can be reconstructed from the vorticity thanks to the Biot-Savart law
\[
u(x, t) = \int_\Omega \nabla_x^\perp G_\Omega(x, y) \omega(y, t) \, dy,
\]
where $G_\Omega$ is the Green's function of $\Omega$ given by
\begin{equation}
G_\Omega(x, y) = \frac{1}{2\pi} \ln |x-y| + \gamma_\Omega(x, y),
\end{equation}
and $\nabla^\perp G:=(-\partial_2G,\partial_1G)$. 
Here $\gamma_\Omega(x, y)$ is a smooth symmetric (potential) function (i.e the smooth part of the Green's function).
In the case of unbounded domain, velocity $u$ should vanish at infinity.
Define also the Robin function:
\[
\tilde{\gamma}_\Omega(x) = \gamma_\Omega(x, x).
\]
Thanks to the Riemann mapping (Chapter 6 \cite{stein2010complex}) and the characterization of automorphisms of the unit disc, a conformal mapping 
$$
\varphi:\Omega\to \mathbb D,\quad\mbox{ biholomorphic and such that for some } x_0\in\Omega^\circ, \varphi(x_0)=0
$$
exists and is unique up to rotations. Thanks to this mapping, one can express the Green's function $G_\Omega$ as 
\[
G_\Omega(x, y) = G_{\mathbb D}(\varphi(x), \varphi(y)),
\]
where the Green's and potential functions for the disc are explicitly given by
\begin{eqnarray}
\label{greens for disc}
G_{\mathbb D}(x, y) = \frac{\ln |x - y|}{2\pi} - \frac{\ln |x - y^*||y|}{2\pi},\quad \gamma_{\mathbb D}(x, y) = -\frac{\ln |x - y^*||y|}{2\pi}.
\end{eqnarray}
Here \( y^* = \frac{y}{|y|^2} \) is the inverse of the point $y$ relative to the unit circle. This allows to re-express Robin's function of a general simply connected domain as
\[
\gamma_\Omega(x, y) = \gamma_{\mathbb D}(\varphi(x), \varphi(y)) + \frac{1}{2\pi} \ln \frac{|\varphi(x) - \varphi(y)|}{|x - y|}.
\]
Now assume that the initial vorticity is concentrated at a single point vortex $\omega_0=\delta_{x_0}$. Then formally the corresponding solution is $\omega(t)=\delta_{x(t)}$, where the motion of $x(t)$ is governed by
\begin{eqnarray*}
    \label{onepntvortex}
x'(t)=\nabla^\perp\tilde\gamma_{\Omega}(x(t)).
\end{eqnarray*}

\begin{Definition}
The point $x_0 \in \Omega$ is called a stationary point if for all time $t$, $x(t)=x_0$.
\end{Definition}
Observe that stationary points are characterized as critical points of the Robin's function. Due to the asymptotic \cite{gustafsson1979motion}
\[
\tilde{\gamma}_\Omega(x) = -\frac{1}{2\pi} \ln d(x, \partial \Omega) + O(1) \quad \text{as } x \to \partial \Omega,
\]
such stationary points always exist. However, their uniqueness strongly depends on additional properties of the geometry of the domain. Stationary points can also be characterized using conformal mappings. Indeed, Gustafsson \cite{gustafsson1979motion} and \cite{donati2021long} proved that $x_0$ is stationary if and only if $\varphi''(x_0)=0$.

In recent years, the longtime behavior of solutions to the 2D Euler equations has regained significant attention mainly in three distinct research directions.

The first and most classical direction initiated by Arnold's \cite{Arnold} involves the search for Lyapunov functionals to study stability in the spirit of Arnold's work. These functionals are crucial tools for understanding the stability properties of solutions and have deep connections to the conservation laws and the geometric structure of the flow.

The second research direction focuses on the development of an understanding of "simple states" that are expected to emerge after a long time. These states include steady states, traveling waves, and time-periodic \cite{garcia2024time} (or quasi-periodic) solutions \cite{berti2023time}. Such solutions often represent the asymptotic behavior of the flow and provide insight into the types of patterns that can persist indefinitely.

The third direction aims to understand the behavior near these "simple states" and the mechanisms that drive perturbations to relax back to the class of "simple states." A major breakthrough in this area was the work of Bedrossian and Masmoudi \cite{BedrossianMasmoudi}, introducing the concept of inviscid damping. More recently, the stability of the twisting properties of Euler's flow near shear flows was established in \cite{drivas2024twisting}. Its possible implications to filamentation or infinite time relaxation still remain open.\\  
At the same time, several works studied the behavior of solutions with concentrated vorticity  around stable stationary points investigating for how long such structures can persist. However, most of the works are based on fluids inside domains with a lot of symmetries including $\mathbb T^2$, $\mathbb R^2$, and the unit disc $\mathbb D$, where the effects of the boundary is either non-existent or ``uniform".
In the absence of the boundary, external forces and viscosity effects, Buttà and Marchioro \cite{M1103725}, analyzed how long a concentrated 
vorticity can be maintained. More precisely, letting $\omega\left(x,t\right)$ be the solution
of $\text{\ensuremath{\left(\ref{eq:sys_euler}\right)}}$ with sign-definite initial vorticity blobs concentrated at $N$ disjoint discs $D(z_j,\varepsilon)$, and as in \cite{donati2021long}, defining
$\tau_{\varepsilon,\beta}$ to be the first time the vorticity support exits a much larger disks of radius $\varepsilon^{\beta}<1$:

\[
\tau_{\varepsilon,\beta}:=\sup\left\{ t\geq0,\ \forall s\in\left[0,t\right],\ \text{supp}\omega\left(\cdot,t\right)\subset\cup_{i=1}^{N}D\left(z_{i}\left(s\right),\varepsilon^{\beta}\right)\right\} .
\]
It is shown in \cite{M1103725,donati2021long} that 
for each $\beta\in\left(0,1/2\right)$ there exist $\varepsilon_{0}>0$
and $\zeta_{0}>0$ such that $\tau_{\varepsilon,\beta}>\zeta_{0}\left|\ln\varepsilon\right|$,
$\forall\varepsilon<\varepsilon_{0}$. In addition, if $\Omega=\mathbb D$ and
$\omega_{0}$ satisfying the condition above with $N=1$ and $z_{1}=0$,
then we can obtain that for every $\beta\in\left(0,1/2\right)$ there
exists $\varepsilon_{0}>0$ and $\alpha>0$ such that $\tau_{\varepsilon,\beta}>\varepsilon^{-\alpha}$,
$\forall\varepsilon<\varepsilon_{0}$. Both estimates were obtained
by Buttà and Marchioro \cite{M1103725}. The first result provides
logarithmic confinement for the entire plane. The second result is
more restrictive, but it is valid for the unit disk with a single
point vortex at its center, and the conclusion is stronger as it yields
a power-law confinement. Later on, Donati and Iftimie \cite{donati2021long}
extended the result to non-disk domains but still with a special hypothesis that the third derivative of the conformal mapping at the stationary point vanishes. This includes domains with $m$-fold symmetric ($m\ge3$) and it is more restrictive than the natural stability condition.
The main purpose of our paper is to study general domains without rotational symmetries.

The point vortex problem is a system of ODEs that was introduced as an approximation of the 2D Euler's system. Helmholtz \cite{helmholtz1858integrale} first introduced the N-point point vortex
system, as a particular solution to $\left(\text{\ref{eq:sys_euler}}\right)$:
\begin{equation}
\forall1\leq k\leq N,\quad\frac{d}{dt}z_{k}\left(t\right)=\sum_{j\neq k}^{N}a_{j}\nabla_{1}^{\bot}G_{\Omega}\left(z_{k}\left(t\right),z_{j}\left(t\right)\right)+a_{k}\frac{1}{2}\nabla^{\bot}\tilde{\gamma}_{\Omega}\left(z_{k}\left(t\right)\right)\label{eq:Kirchhoff-Routh_eq}
\end{equation}
where $\left(a_{k}\right)_{1\leq k\leq N}$ represent the strengths of the
point vortices $\left(z_{k}\left(t\right)\right)_{1\le k\le N}$ and $\tilde{\gamma}_\Omega$
denotes the Robin function of the domain $\Omega$. The equation \eqref{eq:Kirchhoff-Routh_eq} is also called
Kirchhoff-Routh equations \cite{kirchhoff1876vorlesungen,routh1905advanced}.
The point vortex model serves as an approximation for the vortex structure
within an incompressible fluid by discretizing the vortex region into
a finite number of point vortexes, each characterized by a specific
strength (vorticity).

Poincaré \cite{poincare1893theorie} analyzed the dynamics of vortices
between 1891 and 1893, focusing on the motion and interaction of point
vortices. In 1994, Marchioro and Pulvirenti \cite{marchioro2012mathematical}
analyzed the point vortex model in details, including its mathematical
derivation, behavior in different fluid dynamics scenarios, and comparisons
with other fluid dynamics models. They established a rigorous mathematical
foundation for the point vortex model and investigated the behavior
of point vortices in the presence of boundaries. Saffman \cite{saffman1995vortex}
mainly discussed the dynamic behavior and stability of the point vortex
model in 1995. He analyzed the interaction, stability and motion characteristics
of point vortices in fluids, including the formation of vortex streets
and their associated instabilities.

A critical aspect of studying point vortex system is its convergence
toward the global solution of the Euler equations. Essentially, the
question is whether the solution to a point vortex system can become
increasingly close to the global solution of the incompressible Euler
equations as the number of point vortices increases. As $N$ approaches
infinity, if the velocity field of the point vortex system converges
to the weak solution of the Euler equations in the distribution sense,
it is termed weak convergence. Geldhauser and Romito \cite{geldhauser2019point}
demonstrated that under suitable scaling conditions, the velocity
field of the point vortex system can indeed converge weakly to the
Euler equations' weak solution. More information can be found in \cite{liu2001convergence}
and \cite{ceci2021vortex}. Strong convergence, in contrast, necessitates
pointwise convergence within a specific function space. This form
of convergence imposes stricter demands, such as increased smoothness
of the initial conditions and more rigorous control over the distribution
of point vorticities, as detailed in Rosenzweig \cite{rosenzweig2022mean}. 

The point vortex system $\left(\text{\ref{eq:Kirchhoff-Routh_eq}}\right)$
fails when the points $\left(z_{k}\left(t\right)\right)_{1\leq k\leq N}$
are not unique or leave the domain $\Omega$. Although configurations
exist that can lead to point vortex collapse at the interior or the boundary of the domain, they are the exceptions,
see \cite{durr1982vortex}, \cite{marchioro1984vortex}, \cite{marchioro2012mathematical}
for the proof of the case of the ring surfaces, the disk, and the
plane, respectively. Indeed, Donati Martin \cite{Donati2022} shows the global existence of $N$ point vortex system in a bounded 2D domain for almost every initial data.

We are interested in the so-called point vortex confinement problem
with positive vortex strengths $\left(a_{k}\right)_{1\leq k\leq N}>0$.

Assuming that $\left(z_{k}\left(t\right)\right)_{1\leq k\leq N}$
are the solution of $\left(\text{\ref{eq:Kirchhoff-Routh_eq}}\right)$,
based on the above definition of $\tau_{\varepsilon,\beta}$, we define
an exit time $t_{\varepsilon,\beta}$ of the point vortex system:

\begin{align}
t_{\varepsilon,\beta} & :=\sup\left\{ t\geq0,\ \forall s\in\left[0,t\right],\ \left|z_{k}\left(s\right)-x_{0}\right|<\varepsilon^{\beta},\ k=1,\ldots,N\right\} ,\beta<1,\nonumber \\
t_{\varepsilon,1} & :=\sup\left\{ t\geq0,\ \forall s\in\left[0,t\right],\ \left|z_{k}\left(s\right)-x_{0}\right|<\mu\varepsilon,\ k=1,\ldots,N\right\} ,\text{for some constant }\mu>0.\label{def:exit_time}
\end{align}

Donati and Iftimie \cite{donati2021long} established that a single
point vortex is stable under the condition $2|T'(0)|^{3}>|T'''(0)|$.
Essentially, this condition means $B\left(t\right)$ the center of
point vortices is near a rotation when the initial values $\left(z_{k}\left(0\right)\right)_{1\le k\le2}$
are close to $x_0$ (see $\left(\text{\ref{eq:reason_B_rotation}}\right)$).
If $\left|B\right|$ is small enough, the point vortex confinement
becomes the stability analysis near the critical point. We
see that the motion of $z_{1}\left(t\right)$ and $z_{2}\left(t\right)$ can be
controlled by their centre of vorticity $B\left(t\right)$ and their distance  $\xi\left(0\right)$. Indeed, the initial values $B\left(0\right)$
and $\xi\left(0\right)$ are determined by the initial values of
$\left(z_{k}\left(0\right)\right)_{1\le k\le2}$. 

By~\eqref{tilde_quantity}, one might be tempted to apply the Morse lemma near $\xi=\mathbf{0}$.  
However, the system involves four variables and is degenerate at $\xi=\mathbf{0}$, so the Morse lemma is not applicable.  \\
Instead, we expect that the celebrated KAM theory~\cite{Kolmogorov1954} (see also the modern exposition~\cite{Hubbard2004} and the tutorials~\cite{Wayne1994,de2003tutorial}) can be applied to show that, for almost every sufficiently small pair $(\xi,B)$, the solution is quasi-periodic. Rewriting the system in terms of $(\xi,B)$ reveals that the unperturbed Hamiltonian is singular: the determinant of its Hessian matrix vanishes for all $(\xi,B)$. Fortunately, by expanding the full Hamiltonian to higher order, one recovers a non-degenerate unperturbed Hamiltonian, to which KAM theory becomes applicable.  \\
Another interesting feature of our Hamiltonian is that it takes the form
\[
c \ln(|\xi|^2) + \text{(quadratic terms)} + \text{(higher-order terms)}.
\]
The additional logarithmic term introduces a singularity. From the dynamical systems viewpoint, however, this can be beneficial: the logarithm generates a very strong rotational effect, which may dominate and suppress potential instabilities arising from linear or higher-order terms.  \\
To illustrate this mechanism, we present a toy model. Let $x=(x_1,x_2)$ and consider the following ODE:

$$
\dot{x}=\left(\begin{array}{cc}
1 & 0\\
0 & -1
\end{array}\right)x+\frac{x^{\perp}}{\left|x\right|^{2}}.
$$
Its Hamiltonian is $H=\ln(|x|)-x_1x_2$, and the linear part indicates that there is an unstable direction. However, one can apply the Morse Lemma to the conserved quantity 
$$e^{2H}=x_{1}^{2}+x_{2}^{2}+h.o.t$$
and obtain the periodic solution for small initial data.



This paper extends the investigation to the stability of a two-point
vortex system near zero, with the following main result:

\begin{thm}\label{main result}
Given a simply connected domain $\Omega\subsetneq\mathbb{R}^2$ with $0\in \Omega^\circ$ the stationary point.
Let $\varphi:\Omega\to\mathbb D$ be biholomorphic such that $\varphi(0)=0$ and the stability condition 
$$2\left|\varphi^{\prime}(0)\right|^{3}>\left|\varphi^{\prime\prime\prime}(0)\right|
$$
is satisfied. Then there exists $0<\varepsilon_{0}\ll1$, and a constant $\mu>0$ such that
for any $0<\varepsilon\le\varepsilon_{0}$, if the vorticity strengths are such that $a_1+a_2\neq0$, then with the possible
exception of a set of Lebesgue measure zero, for almost every the initial positions of the point vortices satisfying
\begin{equation}
\left|z_{k}\left(0\right)\right|<\varepsilon,\quad k=1,2,\label{eq:initial_value}
\end{equation}
the solutions $z_{k}\left(t\right)$ of \eqref{eq:Kirchhoff-Routh_eq} satisfy $\left|z_{k}\left(t\right)\right|<\mu \varepsilon$, for any $t>0$ and
$k=1,2$ i.e. $t_{\varepsilon, 1}=+\infty$ .
\end{thm}
In the borderline of the stability threshold i.e. when $2\left|\varphi^{\prime}(0)\right|^{3}=\left|\varphi^{\prime\prime\prime}(0)\right|
$, the stability result is weaker as the confinement time becomes a the power law. More precisely,
\begin{thm}\label{main result 2}
Given a simply connected domain $\Omega\subsetneq\mathbb{R}^2$ and $0\in \Omega^\circ$ a stationary point.
Denote by $\varphi:\Omega\to\mathbb D$ a biholomorphic conformal map such that $\varphi(0)=0$ satisfying
$2\left|\varphi^{\prime}(0)\right|^{3} = \left|\varphi^{\prime\prime\prime}(0)\right|.$
Let $z_1$, $z_2$ be two point vortices with vorticity strengths $a_1+a_2\neq0$. Then, for any $\beta<1$, there exists $\varepsilon_0>0$ (small enough) such that for all $0<\varepsilon\le \varepsilon_0$, for any initial positions satisfying
\begin{equation*}
\left|z_{k}\left(0\right)\right|<\varepsilon,\quad k=1,2,
\end{equation*}
the solutions $z_{k}\left(t\right)$ of \eqref{eq:Kirchhoff-Routh_eq} satisfy $|z_k(t)| < \varepsilon^\beta$ for $t<C\varepsilon^{-\alpha}$ with $\alpha = \min(\frac{1-\beta}{2},\frac{1}{4})$ and $k=1,2$ i.e. $t_{\varepsilon,\beta}\ge C\varepsilon^{-\alpha} $.
\end{thm}

\begin{rem}
    \begin{itemize}
        \item Clearly, when the two point vortices have opposite sign, i.e. $a_1a_2<0$, the conservation of the Hamiltonian, guarantees that no collision of the two vertices can happen. However, if they have the same sign, collision may happen only at the boundary (see Theorem 3.3.2 in \cite{newton} as well as Chapter 15 of Flucher \cite{Flucher1999}). Such eventual time is {\it `a priori'} much bigger than the first exit time $t_{\varepsilon,1}$. Since in our case $t_{\varepsilon,1}=\infty$, then no collision also can happen in this case. 
        \item Theorem \ref{main result} gives the first result about the longtime confinement for two point vortices in a genernal simply connected domain. 

\item Our result proves the all time confinement, and the fact that $\beta=1$ aligns our result with the classical definition of stability in the sense of Lyapunov. 
\item The presence of the boundary of the domain breaks, in general\footnote{if the domain does not enjoy further symmetries like rotationally invariant}, the integrability of the two point vortices system (see ref \cite{newton}).
The example of the two-point vortex model given by Luithardt, Kadtke, and Pedrizzetti \cite{luithardt1994chaotic} demonstrated the chaos feature phenomenon without explicit disturbance and revealed the non-integrability of the system. On the other hand the N point vortices system in $\mathbb{R}^2$ has three independent conserved quantities, coming from the conservation of interaction energy, rotational symmetry and translation symmetry. This explains why the case of $N\ge 4$ is generally not integrable. In general, involving the boundary breaks the rotational and translation symmetry and only one conserved quantity remains which partially shows that only the case $N=1$ with a general boundary is integrable.
\item As it is stated in Thoerem \ref{main result}, not all small initial positions will lead to all time confinement. However the set of the exceptions is of Lebesgue measure zero. Indeed, accordinig to the proof, given the conformal map $\varphi$ and initial positions $z_i(0),i=1,2$ the algorithm to verify whether $z_i(0)$ meet the confinement condtion is the following:
\begin{itemize}
    \item Calculating $\xi(0)=\frac{\sqrt{a_1 a_2}}{a_1+a_2}(z_1(0)-z_2(0)),B(0)=\frac{a_1 z_1(0)+a_2 z_2(0)}{a_1+a_2}$ and the action variable $I_\xi(0) = \frac{|\xi(0)|^2}{2}, I_B(0)=\frac{|B(0)|^2}{2}$, where $a_1, a_2$ are positive. Other cases are similar.
    \item Extracting the part of the Hamiltonian \eqref{eqn:new_H} that depends only on the action variables $I_\xi, I_B$: $h(I_\xi, I_B) = \frac{1}{4\pi^2}\int_0^{2\pi}\int_0^{2\pi} H(I_\xi,I_B,\theta_\xi,\theta_B)d\theta_\xi d\theta_B$.
    \item Calculating the frequency vector $\omega^* = \nabla h(I_\xi(0), I_B(0))$.
    \item Verifying that whether $\omega^*$ belongs the following Cantor-like set:
    \[
    \mathcal{D} = \left\{\omega\in\mathbb{R}^2 \bigg| \text{for some $C, \nu >0$, } |\omega\cdot k| \ge \frac{C}{|k|^\nu},\ \forall k\in \mathbb{Z}^2\setminus {0}  \right\}.
    \]
    If $\omega^*\in \mathcal{D}$ which is exactly that $\omega^*$ satisfies the Diophantine condition, then the all time confinement holds. 
\end{itemize}
\item {In the degenerate case $a_1+a_2=0$, the longtime confinement property may be lost as shown by the following example. 
Consider the case $\Omega=\mathbb{D}$, $a_1=-a_2=1$. Furthermore, without loss of generality, one can assume $z_1 = \bar{z}_2 := x_1 + ix_2$. Indeed, for the general initial data $|z_1(0)| = |z_2(0)|$, there exists a rotation $\lambda$ such that $\lambda z_1(0) = \overline{\lambda z}_2(0)$. The system \eqref{eq:Kirchhoff-Routh_eq} then reduces to 
\begin{align*}
2\pi\dot{x}_1 & =\frac{-x_2}{1-x^{2}_1-x_2^{2}}-\frac{x_2\left(1+x^{2}_1+x_2^{2}\right)}{1-2x^{2}_1+2x_2^{2}+\left(x^{2}_1+x_2^{2}\right)^{2}}+\frac{1}{2x_2},\\
2\pi\dot{x}_2 & =\frac{x_1}{1-x^{2}_1-x_2^{2}}-\frac{x_1\left(1-x^{2}_1-x_2^{2}\right)}{1-2x^{2}_1+2x_2^{2}+\left(x^{2}_1+x_2^{2}\right)^{2}},
\end{align*}
where the trajectory is defined by 
\[
\frac{|1-z^2|}{(1-|z|^2)|z-\bar{z}|}=C.
\]
For $|z(0)|$ small, one can plot its trajectory (see Figure \ref{fig:degenerate_orbit}) and observe that it loses the confinement. In fact, by the equation of $\dot{x}_1$, for small $|z|$ the term $\frac{1}{2x_2}$ is big and the vortex will leave $\varepsilon^\beta$ disk very soon. }
\begin{figure}[ht]
\centering
\includegraphics[scale=0.5]{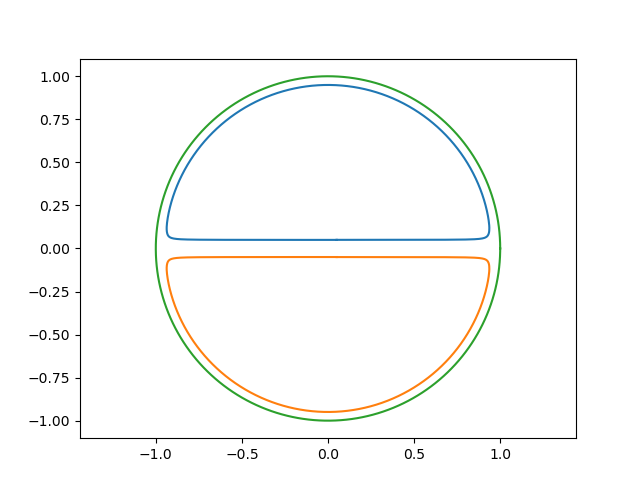}
\caption{The trajectory of $z_{1},z_{2}$ for $\Omega=\mathbb{D},a_{1}=-a_{2}=1$. \label{fig:degenerate_orbit}}
\end{figure}
    \end{itemize}
\end{rem}

The structure of this paper is outlined as follows: The subsequent
section introduces the necessary notation and lists some key lemmas.
Section 3 presents the proof of Theorem \ref{main result}
and examines the limitations of this method when applied to the cases
where $N\geq3$. Finally, the appendix provides the relevant content
on Cauchy's estimates.

\section{Preliminary}

We list some useful expansions of the regular part of the Green's function and some preparation work in this section for Hamiltonian reduction. 
Recall that for vector function $f=(f_1,f_2)$, $f^\perp=(-f_2,f_1)=if$ in the complex planar. Then the singular part of $\nabla_{x}^{\bot}G_{\Omega}\left(z_{k}\left(t\right),z_{j}\left(t\right)\right)$
is given by $\frac{i\left(z_{k}\left(t\right)-z_{j}\left(t\right)\right)}{\left|z_{k}\left(t\right)-z_{j}\left(t\right)\right|^{2}}$. And we can
rewrite ($\text{\ref{eq:Kirchhoff-Routh_eq}}$) as:
\begin{align}
\frac{d}{dt}z_{k}\left(t\right)= & a_{k}\frac{i}{2}\nabla\tilde{\gamma}_{\Omega}\left(z_{k}\left(t\right)\right)\nonumber \\
 & +\sum_{j\neq k}^{N}a_{j}\left(\frac{1}{2\pi}\frac{i\left(z_{k}\left(t\right)-z_{j}\left(t\right)\right)}{\left|z_{k}\left(t\right)-z_{j}\left(t\right)\right|^{2}}+i\nabla_{1}\gamma_\Omega\left(z_{k}\left(t\right),z_{j}\left(t\right)\right)\right).\label{eq:point_vortex}
\end{align}
We intend to push Taylor series expansion of the
regular part of Green's function to higher order to obtain a more intuitive
understanding of the point vortex equation $\left(\text{\ref{eq:point_vortex}}\right)$. 

Recall that Buttà and Donati \cite{donati2021long} proved that a
single point vortex placed at $x_{0}$ being stationary is equivalent
to $\varphi''\left(x_{0}\right)=0$ and they showed 
\begin{equation}
\nabla_{1}\gamma_\Omega\left(x,y\right)=\frac{\overline{\varphi'\left(x\right)}}{2\pi\overline{\left(\varphi\left(x\right)-\varphi\left(y\right)\right)}}-\frac{\overline{\varphi'\left(x\right)}}{2\pi\overline{\left(\varphi\left(x\right)-\varphi\left(y\right)^{*}\right)}}-\frac{1}{2\pi\overline{\left(x-y\right)}}\label{eq:expand_gamma}
\end{equation}
where $\varphi\left(y\right)^{*}=\frac{1}{\overline{\varphi\left(y\right)}}$
is the circle inversion and $\nabla_1$ is the gradient with respect to the first variable.

Without loss of generality, we can set the stationary point at $x_{0}=0$
and have the following expansions of $\gamma_\Omega$. This will help us to describe the lowest order terms in the Hamiltonian.
\begin{lem}
\label{lem:gamma}For any $|x|$, $|y|\leq\delta$
with $\delta$ small enough, it holds that

\begin{equation}
\label{lem:2.1}
\gamma_{\Omega}\left(x,y\right) =\frac{1}{2\pi}\ln\left|\varphi'\left(0\right)\right|+\frac{1}{2}\Re\left(c_{0}\left(x^{2}+xy+y^{2}\right)\right)+c_{1}\Re\left(x\bar{y}\right)+E_\Omega\left(x,y\right),
\end{equation}

where the remainder is controlled by 
\[
\left|E_{\gamma_\Omega}\left(x,y\right)\right|\leq C\delta^{3}
\]
 with $C$ being a constant that only depends on $\Omega$ and \[c_{0}=\frac{\varphi'''\left(0\right)}{6\pi\varphi'\left(0\right)},c_{1}=\frac{\left|\varphi'\left(0\right)\right|^{2}}{2\pi}.\]
\end{lem}
\begin{proof}
The proof is an implication of \eqref{greens for disc}
and Proposition $2.1$ in \cite{donati2021long}. We know that
\begin{align}
\gamma_\Omega\left(x,y\right)= & \gamma_{\mathbb D}\left(\varphi\left(x\right),\varphi\left(y\right)\right)+\frac{1}{2\pi}\ln\frac{\left|\varphi\left(x\right)-\varphi\left(y\right)\right|}{\left|x-y\right|}\nonumber \\
 & =\frac{1}{2\pi}\ln\frac{\left|\varphi\left(x\right)-\varphi\left(y\right)\right|}{\left|x-y\right|}-\frac{1}{2\pi}\ln\left|1-\varphi\left(x\right)\overline{\varphi\left(y\right)}\right|.\label{eq:=000020gamma}
\end{align}
We also know that 
$$
\frac{1}{2\pi}\ln\frac{\left|\varphi\left(x\right)-\varphi\left(y\right)\right|}{\left|x-y\right|} = \Re\left(\frac{1}{2\pi}\ln\frac{\varphi\left(x\right)-\varphi\left(y\right)}{x-y}\right).
$$
Then expanding the analytic part around $\left(0,0\right)$ up to the
second order yields 
\begin{equation}
\frac{1}{2\pi}\ln\frac{\varphi\left(x\right)-\varphi\left(y\right)}{x-y}=\frac{\ln\varphi'\left(0\right)}{2\pi}+\frac{\varphi'''\left(0\right)}{12\pi\varphi'\left(0\right)}xy+\frac{\varphi'''\left(0\right)}{12\pi\varphi'\left(0\right)}\left(x^2+y^2\right)+E\left(x,y\right),\label{eq:first_term_gamma}
\end{equation}
where $\varphi''(0)=0$ is applied since we assume $0$ is the stationary point so that the linar term vanishes and $E$ is the remainder. Let $F\left(x,y\right)=\frac{1}{2\pi}\ln\frac{\varphi\left(x\right)-\varphi\left(y\right)}{x-y}$. Noting that $\varphi'(0)\neq0$ since $\varphi$ is conformal and the fact that $\lim_{y\to x}\frac{\varphi\left(x\right)-\varphi\left(y\right)}{x-y}=\varphi'(x)$, it is clear to see that $F(x,y)$ is holomorphic when $x, y$ are close to $0$.
Moreover, noting that 
\[
E\left(x,y\right)=\frac{\pi}{2}\sum_{\alpha_{1}+\alpha_{2}\geq 3}\frac{\partial_{x}^{\alpha_{1}}\partial_{y}^{\alpha_{2}}F\left(0,0\right)}{\alpha_{1}!\alpha_{2}!}x^{\alpha_{1}}y^{\alpha_{2}},
\]
let $r=\min_{x\in\partial \Omega}\left|x\right|$ so that $D_r=\left\{z\in\mathbb{C}:\; \left|z\right|\le r\right\}\subset\Omega$. Define $M=\max_{x,y\in\partial D_r}\left|\ln\frac{\varphi\left(x\right)-\varphi\left(y\right)}{x-y}\right|$. $M$ is finite since $D_r$ is bounded and obviously $M, r, D_r$ are fixed once the domain $\Omega$ is fixed. Thanks to Cauchy's estimate $\text{\ref{prop:cauchy=000020estimate}}$
the derivatives of $F\left(x,y\right)$ are controlled by
\[
\left|E\left(x,y\right)\right|\leq \frac{\pi M}{2}\sum_{\alpha_{1}+\alpha_{2}\geq 3}\left|\frac{x}{r}\right|^{\alpha_{1}}\left|\frac{y}{r}\right|^{\alpha_{2}}\leq \frac{\pi M}{2}\sum_{k\geq 3}\left|\frac{\delta}{r}\right|^{k}= \frac{\pi M}{2}\frac{\delta^3}{r^2\left(r-\delta\right)}.
\]
By choosing $\delta < \frac{r}{2}$, $E$ is controlled by 
\[
\left|E\left(x,y\right)\right| \le C\delta^{3},
\]
where $C$ is a constant only depending on $\Omega$. Therefore letting $c_0=\frac{\varphi'''\left(0\right)}{6\pi\varphi'\left(0\right)}$, we have 
$$
\frac{1}{2\pi}\ln\frac{\left|\varphi\left(x\right)-\varphi\left(y\right)\right|}{\left|x-y\right|} = \frac{\ln|\varphi'(0)|}{2\pi} +\frac{c_0}{4}xy+\frac{\bar{c}_0}{4}\bar{x}\bar{y}+\frac{c_0}{4}\left(x^2+y^2\right)+\frac{\bar{c}_0}{4}\left(\bar{x}^2+\bar{y}^2\right)+ \Re\left(E(x,y)\right),
$$
where $|\Re\left(E(x,y)\right)| \le C\delta^3 .$
Similarly, by expanding $\varphi\left(x\right)$ up to zero order and applying
Proposition \ref{prop:cauchy=000020estimate} gives
\begin{equation}\label{eq:phi_estimate}
    \varphi\left(x\right)=\varphi'\left(0\right)x+R_{\varphi,3}\left(x\right),\left|R_{\varphi,3}\left(x\right)\right|\leq M_{\varphi}\frac{\delta^3}{r^2\left(r-\delta\right)}\le\frac{\delta^3}{r^2\left(r-\delta\right)}\leq C\delta^3
\end{equation}
with $M_{\varphi}=\max_{x\in\partial D_r}\left|\varphi\left(x\right)\right|\le 1$ and $r<\frac{\delta}{2}$.
Clearly, similar method can not be applied to $\ln\left|1-\varphi\left(x\right)\overline{\varphi\left(y\right)}\right|$
directly since it is not holomorphic. Instead we write 
\[
\frac{1}{2\pi}\ln\left|1-\varphi\left(x\right)\overline{\varphi\left(y\right)}\right|=-\frac{1}{4\pi}\sum_{n=1}^{\infty}\frac{\left(\varphi\left(x\right)\overline{\varphi\left(y\right)}+\overline{\varphi\left(x\right)}\varphi\left(y\right)-\left|\varphi\left(x\right)\right|^{2}\left|\varphi\left(y\right)\right|^{2}\right)^{n}}{n}.  
\]
Let $G\left(x,y\right)=\varphi\left(x\right)\overline{\varphi\left(y\right)}+\overline{\varphi\left(x\right)}\varphi\left(y\right)-\left|\varphi\left(x\right)\right|^{2}\left|\varphi\left(y\right)\right|^{2}$. Then we have 
\begin{align}
    \frac{1}{2\pi}\ln\left|1-\varphi\left(x\right)\overline{\varphi\left(y\right)}\right|& =-\frac{1}{4\pi}\left(\varphi\left(x\right)\overline{\varphi\left(y\right)}+\overline{\varphi\left(x\right)}\varphi\left(y\right)\right)+\left(\frac{1}{4\pi}\left|\varphi\left(x\right)\right|^{2}\left|\varphi\left(y\right)\right|^{2}-\frac{1}{4\pi}\sum_{n=2}^{\infty}\frac{G\left(x,y\right)^n}{n}\right)  \nonumber \\
    & = -\frac{|\varphi'\left(0\right)|^2}{4\pi}\left(x\bar{y}+\bar{x}y\right)+R_{4}\left(x,y\right), \label{eq:last_term_gamma}
\end{align}
where $R_{4}\left(x,y\right)$ is composed of higher-order terms in $\varphi\left(x\right)\overline{\varphi\left(y\right)}+\overline{\varphi\left(x\right)}\varphi\left(y\right)$,  $\left|\varphi\left(x\right)\right|^{2}\left|\varphi\left(y\right)\right|^{2}$ and $\sum_{n=2}^{\infty}\frac{G\left(x,y\right)^n}{n}$.

With \eqref{eq:phi_estimate} we have
\[
\left|G\left(x,y\right)\right|\leq2\left|R_{\varphi,0}\left(x\right)R_{\varphi,0}\left(y\right)\right|+\left|R_{\varphi,0}^{2}\left(x\right)R_{\varphi,0}^{2}\left(y\right)\right|\le 2C\delta^2+C\delta^4\le C\delta^2,
\]
and then
\begin{equation}
|R_{4}\left(x,y\right)|\le C\delta^4.    
\end{equation}

Denoting $R_{4}\left(x,y\right) + E\left(x,y\right)$
by $E_{\gamma_\Omega}\left(x,y\right)$ and collecting $\left(\text{\ref{eq:first_term_gamma}}\right)$
and $\left(\text{\ref{eq:last_term_gamma}}\right)$ proves this lemma.
\end{proof}
Next lemma shows the expansion of derivative of $\gamma_\Omega$ up to first order, leading to the expansion of the equations of motions.
\begin{lem}
\label{lem:grad_gamma}For any $\left|x\right|$, $\left|y\right|\leq\delta$,
with $\delta$ small enough, we know
\[
\nabla_{1}\gamma_\Omega\left(x,y\right)=\bar{c}_{0}\bar{x}+\frac{1}{2}\bar{c}_{0}\bar{y}+c_{1}y+E_{\nabla_1\gamma_\Omega}\left(x,y\right),
\]
\[
\nabla\tilde{\gamma}_{\Omega}\left(x\right)=3\bar{c}_{0}\bar{x}+2c_{1}x+E_{\nabla\tilde{\gamma}_\Omega}\left(x\right),
\]
where 
\[
c_{0}=\frac{\varphi'''(0)}{6\pi \varphi'(0)},\quad c_{1}=\frac{|\varphi'(0)|^{2}}{2\pi}.
\]
The remainders are controlled by 
\[
\left|E_{\nabla_1\gamma_\Omega}\left(x,y\right)\right|\le C\delta^2, \left|E_{\nabla\tilde{\gamma}_\Omega}\left(x\right)\right|\le C\delta^2.
\]
\end{lem}
\begin{proof}
Conjugating $\left(\text{\ref{eq:expand_gamma}}\right)$ we have
\begin{equation}
\overline{\nabla_{1}\gamma_{\Omega}\left(x,y\right)}=\frac{\varphi'\left(x\right)}{2\pi\left(\varphi\left(x\right)-\varphi\left(y\right)\right)}-\frac{1}{2\pi\left(x-y\right)}-\frac{\varphi'\left(x\right)}{2\pi\left(\varphi\left(x\right)-\varphi\left(y\right)^{*}\right)}.\label{eq:expand_grad_gamma}
\end{equation}
Due to the fact that $\left(\text{\ref{eq:expand_gamma}}\right)$
has no interior singularity, $\frac{\varphi'\left(x\right)}{2\pi\left(\varphi\left(x\right)-\varphi\left(y\right)\right)}-\frac{1}{2\pi\left(x-y\right)}$ is analytic near $\left(x,y\right)=\left(0,0\right)$. Similar to the proof of the previous lemma, let $F\left(x,y\right)=\ln\frac{\left|\varphi\left(x\right)-\varphi\left(y\right)\right|}{\left|x-y\right|}$ then it is easy to see that
\[
\partial_xF = \frac{\varphi'\left(x\right)}{\varphi\left(x\right)-\varphi\left(y\right)}-\frac{1}{x-y}.
\]
Expanding $\partial_x F$ around $\left(x,y\right)=\left(0,0\right)$ up to the first order 
yields
\begin{align}
\frac{\varphi'\left(x\right)}{\varphi\left(x\right)-\varphi\left(y\right)}-\frac{1}{x-y} & =\lim_{\left(x,y\right)\to\left(0,0\right)}\left(\frac{\varphi'\left(x\right)}{\varphi\left(x\right)-\varphi\left(y\right)}-\frac{1}{x-y}\right)\label{eq:T_expand}\\
 & \qquad+x\lim_{\left(x,y\right)\to\left(0,0\right)}\nabla_{1}\left(\frac{\varphi'\left(x\right)}{\varphi\left(x\right)-\varphi\left(y\right)}-\frac{1}{x-y}\right)\nonumber \\
 & \qquad+y\lim_{\left(x,y\right)\to\left(0,0\right)}\nabla_{2}\left(\frac{\varphi'\left(x\right)}{\varphi\left(x\right)-\varphi\left(y\right)}-\frac{1}{x-y}\right)\nonumber \\
 & \qquad+\sum_{\alpha_1+\alpha_2\ge 2}\frac{\partial_x^{\alpha_1+1}\partial_y^{\alpha_2} F\left(0,0\right)}{\alpha_1!\alpha_2!}x^{\alpha_1}y^{\alpha_2}.\nonumber 
\end{align}
The zero order term can be calculated as following 
\begin{align}
\lim_{\left(x,y\right)\to\left(0,0\right)}\left(\frac{\varphi'\left(x\right)}{\varphi\left(x\right)-\varphi\left(y\right)}-\frac{1}{x-y}\right) & =\lim_{x\to0}\lim_{y\to x}\left(\frac{\varphi'\left(x\right)}{\varphi\left(x\right)-\varphi\left(y\right)}-\frac{1}{x-y}\right)\nonumber \\
 & =\lim_{x\to0}\lim_{y\to x}\frac{\varphi'\left(x\right)-\frac{\varphi\left(x\right)-\varphi\left(y\right)}{x-y}}{\left(x-y\right)\frac{\varphi\left(x\right)-\varphi\left(y\right)}{x-y}}\nonumber \\
 & =\lim_{x\to0}\frac{\varphi''\left(x\right)}{2\varphi'\left(x\right)}=0,\label{eq:0_order}
\end{align}
where we used the fact $\varphi''\left(0\right)=0$. For the first order
terms, the calculation is similar:
\begin{align}
 & \lim_{\left(x,y\right)\to\left(0,0\right)}\nabla_{1}\left(\frac{\varphi'\left(x\right)}{\varphi\left(x\right)-\varphi\left(y\right)}-\frac{1}{x-y}\right)\nonumber \\
= & \lim_{x\to0}\lim_{y\to x}\frac{\left(\varphi\left(x\right)-\varphi\left(y\right)\right)\varphi''\left(x\right)-\varphi'\left(x\right)^{2}}{\left(\varphi\left(x\right)-\varphi\left(y\right)\right)^{2}}+\frac{1}{\left(x-y\right)^{2}}\nonumber \\
= & \lim_{x\to0}\lim_{y\to x}\frac{\frac{\left(\varphi\left(x\right)-\varphi\left(y\right)\right)\varphi''\left(x\right)-\varphi'\left(x\right)^{2}+\left(\varphi\left(x\right)-\varphi\left(y\right)\right)^{2}/\left(x-y\right)^{2}}{\left(x-y\right)^{2}}}{\left(\frac{\varphi\left(x\right)-\varphi\left(y\right)}{x-y}\right)^{2}}\nonumber \\
= & \lim_{x\to0}\frac{\varphi'\left(x\right)\varphi'''\left(x\right)/3-\varphi''\left(x\right)/4}{\varphi'\left(x\right)^{2}}=\frac{\varphi'''\left(0\right)}{3\varphi'\left(0\right)},\label{eq:1st_order_x}
\end{align}
and 
\begin{align}
 & \lim_{\left(x,y\right)\to\left(0,0\right)}\nabla_{2}\left(\frac{\varphi'\left(x\right)}{\varphi\left(x\right)-\varphi\left(y\right)}-\frac{1}{x-y}\right)\nonumber \\
= & \lim_{x\to0}\lim_{y\to x}\frac{\varphi'\left(x\right)\varphi'\left(y\right)}{\left(\varphi\left(x\right)-\varphi\left(y\right)\right)^{2}}-\frac{1}{\left(x-y\right)^{2}}\nonumber \\
= & \lim_{x\to0}\lim_{y\to x}\frac{\frac{\varphi'\left(x\right)\varphi'\left(y\right)-\left(\varphi\left(x\right)-\varphi\left(y\right)\right)^{2}/\left(x-y\right)^{2}}{\left(x-y\right)^{2}}}{\left(\frac{\varphi\left(x\right)-\varphi\left(y\right)}{x-y}\right)^{2}}\nonumber \\
= & \lim_{x\to0}\frac{\varphi'\left(x\right)\varphi'''\left(x\right)/6-\varphi''\left(x\right)/4}{\varphi'\left(x\right)^{2}}=\frac{\varphi'''\left(0\right)}{6\varphi'\left(0\right)}.\label{eq:1st_order_y}
\end{align}
The remainder term in \eqref{eq:T_expand} is controlled by Cauchy's estimate (see Proposition \ref{prop:cauchy=000020estimate}): 
\begin{equation}
    \label{eq:remainder_estimate}
    \sum_{\alpha_1+\alpha_2\ge 2}\frac{\partial_x^{\alpha_1+1}\partial_y^{\alpha_2} F\left(0,0\right)}{\alpha_1!\alpha_2!}x^{\alpha_1}y^{\alpha_2} \le \max_{x,y\in\partial D_r}\left|\partial_xF\right|\sum_{k\geq2}\left|\frac{\delta}{r}\right|^{k}\le C\delta^2,
\end{equation}
where $\delta\ll r:=\min_{x\in\partial\Omega}|x|$ is applied. Plugging \eqref{eq:0_order}, \eqref{eq:1st_order_x}, \eqref{eq:1st_order_y} and \eqref{eq:remainder_estimate} into \eqref{eq:T_expand} yields
\begin{equation}
\frac{\varphi'\left(x\right)}{\varphi\left(x\right)-\varphi\left(y\right)}-\frac{1}{x-y}=\frac{\varphi'''\left(0\right)}{3\varphi'\left(0\right)}x+\frac{\varphi'''\left(0\right)}{6\varphi'\left(0\right)}y+E_1\left(x,y\right), \left|E_1\left(x,y\right)\right|\le C\delta^2.\label{eq:9}
\end{equation}
For the last term we have
\begin{align}
\frac{\varphi'\left(x\right)}{\varphi\left(x\right)-\varphi\left(y\right)^{*}}= & \frac{\varphi'\left(x\right)\overline{\varphi\left(y\right)}}{\varphi\left(x\right)\overline{\varphi\left(y\right)}-1}\nonumber\\
 =& -\varphi'\left(x\right)\overline{\varphi\left(y\right)}\left(1+R_1\left(\varphi\left(x\right)\overline{\varphi\left(y\right)}\right)\right),\label{eq:dgamma_estimate_2}
\end{align}
where $R_1\left(z\right)$ satisfying $\frac{1}{1-z}=1+R_1\left(z\right)$ is controlled by $\left|R_1\left(z\right)\right|\le C\left|z\right|$. 
We only need to extract the first order terms in $x$ and $y$. Expanding
$\varphi'\left(x\right), \varphi\left(x\right)$ around $x=0$ and with the help of Proposition \ref{prop:cauchy=000020estimate} yield for $\left|x\right|<\delta$
\begin{equation}
    \label{eq:phi_prime_expand}
    \varphi'\left(x\right) = \varphi'\left(0\right) + E_{\varphi'}\left(x\right), |E_{\varphi'}\left(x\right)|\le C\delta,
\end{equation}
and 
\begin{equation}
    \label{eq:phi_expand}
    \varphi\left(x\right) = \varphi'\left(0\right) + E_\varphi\left(x\right), |E_\varphi\left(x\right)|\le C\delta^2.
\end{equation}
Plugging \eqref{eq:phi_prime_expand} and \eqref{eq:phi_expand} into \eqref{eq:dgamma_estimate_2} gives
\begin{equation}
\frac{\varphi'\left(x\right)}{\varphi\left(x\right)-\varphi\left(y\right)^{*}}=-\left|\varphi'\left(0\right)\right|^{2}\bar{y}+E_2\left(x,y\right), \left|E_2\left(x,y\right)\right|\le C\delta^2. \label{eq:10}
\end{equation}
Collecting $\left(\text{\ref{eq:9}}\right)$ and $\left(\text{\ref{eq:10}}\right)$
proves the first expansion.

To prove the second one, the symmetric property of $\gamma\left(x,y\right)$
gives
\[
\nabla_{1}\gamma_{\Omega}\left(x,y\right)=\nabla_{2}\gamma_{\Omega}\left(y,x\right).
\]
So by the definition of Robin's function
\[
\nabla\tilde{\gamma}_{\Omega}\left(x\right)=\nabla_{1}\gamma_{\Omega}\left(x,x\right)+\nabla_{2}\gamma_{\Omega}\left(x,x\right)=2\nabla_{1}\gamma_{\Omega}\left(x,x\right),
\]
the proof is finished. 
\end{proof}

\section{Proof of the main Theorems}
Before stating the proof, we give a clear definition of the center
of a vortex in the case of $n=2$
\begin{equation}
B\left(t\right):=\frac{\sum_{k=1}^{2}a_{k}z_{k}\left(t\right)}{\sum_{k=1}^{2}a_{k}}=\frac{a_{1}z_{1}\left(t\right)+a_{2}z_{2}\left(t\right)}{a_{1}+a_{2}}.\label{eq:center_vorticity}
\end{equation}
The full system of $N=2$ reads
\begin{equation}
\begin{array}{c}
\frac{d}{dt}z_{1}=\frac{i}{2}a_{1}\nabla\tilde{\gamma}_{\Omega}\left(z_{1}\right)+ia_{2}\frac{z_{1}-z_{2}}{2\pi|z_{1}-z_{2}|^{2}}+ia_{2}\nabla_{1}\gamma_\Omega\left(z_{1},z_{2}\right),\\
\frac{d}{dt}z_{2}=\frac{i}{2}a_{2}\nabla\tilde{\gamma}_{\Omega}\left(z_{2}\right)+ia_{1}\frac{z_{2}-z_{1}}{2\pi|z_{2}-z_{1}|^{2}}+ia_{1}\nabla_{1}\gamma_\Omega\left(z_{2},z_{1}\right).
\end{array}\label{eq:full_sys}
\end{equation}
The total energy $H$ is conserved in time:
\begin{equation}
H=\frac{1}{2}\sum_{i=1}^{2}a_{i}^{2}\gamma_\Omega\left(z_{i},z_{i}\right)+a_{1}a_{2}\left(\gamma_\Omega\left(z_{1},z_{2}\right)+\frac{1}{4\pi}\ln\left(\left|z_{1}-z_{2}\right|^{2}\right)\right).\label{eq:hamiltonian}
\end{equation}
We will work on the system of center of vortexes $B$ and distance between two points $\xi:=\frac{\sqrt{a_1a_2}}{a_1+a_2}(z_1-z_2)$. All the proofs are done for positive $a_1, a_2$. Other cases are actually similar. It is straightforward to obtain the governing equation for $B,\xi$:
\begin{equation}
\begin{cases}
\frac{d}{dt}B=\frac{i}{a_{1}+a_{2}}\left(\frac{a_{1}^{2}}{2}\nabla\tilde{\gamma}_{\Omega}\left(z_{1}\right)+\frac{a_{2}^{2}}{2}\nabla\tilde{\gamma}_{\Omega}\left(z_{2}\right)+a_{1}a_{2}\left(\nabla_{1}\gamma_{\Omega}\left(z_{1},z_{2}\right)+\nabla_{1}\gamma_{\Omega}\left(z_{2},z_{1}\right)\right)\right).\\
\frac{d}{dt}\xi=\frac{ia_{1}a_{2}}{2\pi\left(a_{1}+a_{2}\right)}\frac{\xi}{\left|\xi\right|^{2}}+\frac{i\sqrt{a_{1}a_{2}}}{a_{1}+a_{2}}\left(\frac{a_{1}}{2}\nabla\tilde{\gamma}_{\Omega}\left(z_{1}\right)-\frac{a_{2}}{2}\nabla\tilde{\gamma}_{\Omega}\left(z_{2}\right)+a_{2}\nabla_{1}\gamma_{\Omega}\left(z_{1},z_{2}\right)-a_{1}\nabla_{1}\gamma_{\Omega}\left(z_{2},z_{1}\right)\right).
\end{cases}
\label{eq:B_xi_sys}
\end{equation}
where $z_1,z_2$ are considered as linear combinations of $B,\xi$:
\[
z_{1}=B+\frac{a_{2}}{\sqrt{a_{1}a_{2}}}\xi,\quad z_{2}=B-\frac{a_{1}}{\sqrt{a_{1}a_{2}}}\xi.
\]
Moreover, the system \eqref{eq:B_xi_sys} is actually a Hamiltonian system with 
\begin{equation}\label{eqn:new_H}
H=\frac{a_{1}a_{2}}{2\pi\left(a_{1}+a_{2}\right)}\ln\left|\xi\right|+\frac{1}{a_{1}+a_{2}}\left(\frac{a_{1}^{2}}{2}\tilde{\gamma}_{\Omega}\left(z_{1}\right)+\frac{a_{2}^{2}}{2}\tilde{\gamma}_{\Omega}\left(z_{2}\right)+a_{1}a_{2}\gamma_{\Omega}\left(z_{1},z_{2}\right)\right).
\end{equation}

\subsection{Proof of Theorem \ref{main result}}
Let $a=a_1+a_2$. Applying Lemma \ref{lem:grad_gamma}, one can explicitly write down the linear part of $\dot{B}$ in vector form
\[
\frac{d}{dt}B=a\left(\begin{array}{cc}
\frac{3}{2}c_{0}^{i} & \frac{3}{2}c_{0}^{r}-c_1\\
\frac{3}{2}c_{0}^{r}+c_1 & -\frac{3}{2}c_{0}^{i}
\end{array}\right)B,
\]
where $c_0^r,c_0^i$ are the real and imaginary parts of $c_0$. The eigenvalues of its linear operator is $\pm a\sqrt{c_{1}^{2}-\frac{9}{4}\left|c_{0}\right|^{2}}$. So under the stability condition, the eigenvalues of the linear operator are purely imaginary. Thus the linear operator is a rotation with characteristic frequency $\omega = a\sqrt{c_{1}^{2}-\frac{9}{4}\left|c_{0}\right|^{2}}:=a\omega_c$. Then there exists a canonical linear
transform $Lb=B$ such that at the linear level, $\frac{d}{dt}B$ becomes
\[
\dot{b}=-i\omega b.
\]
And the Hamiltonian can be rewrite as 
\begin{align}
H & =\frac{a_{1}a_{2}}{2\pi a}\ln\left|\xi\right|+\frac{1}{a}\left(\frac{a_{1}^{2}}{2}\tilde{\gamma}_{\Omega}\left(z_{1}\right)+\frac{a_{2}^{2}}{2}\tilde{\gamma}_{\Omega}\left(z_{2}\right)+a_{1}a_{2}\gamma_{\Omega}\left(z_{1},z_{2}\right)\right) \nonumber \\
 & =C+\frac{a_{1}a_{2}}{4\pi a}\ln\frac{\left|\xi\right|^{2}}{2}+\frac{\omega}{2}\left|b\right|^{2}+\frac{a_{1}+a_{2}}{2}\left(c_{0}^{r}\left(\xi_{1}^{2}-\xi_{2}^{2}\right)-2c_{0}^{i}\xi_{1}\xi_{2}\right)+h.o.t,
 \label{eq:H_expand_2nd}
\end{align}
where we applying Lemma \ref{lem:gamma} to expand $\gamma_\Omega$ and $\tilde{\gamma}_\Omega$. The naive choice of the unperturbed Hamiltonian is 
\[
\frac{a_{1}a_{2}}{4\pi a}\ln\frac{\left|\xi\right|^{2}}{2}+\frac{\omega}{2}\left|b\right|^{2}.
\]
However by defining the action angle variable $I=(I_\xi, I_b), \theta=(\theta_\xi,\theta_b)$ satisfying $\xi=\sqrt{2I_\xi}_\xi e^{i\theta_\xi}, b=\sqrt{2I_b} e^{i\theta_b}$, it becomes
\[
\frac{a_{1}a_{2}}{4\pi a}\ln I_\xi+\omega I_b,
\]
whose Hessian matrix is degenerate at any $(I_\xi, I_b)$. Thus we need to add higher order terms into the naive unperturbed Hamiltonian. Because $\varphi(x)$ is analytic for $x\in\Omega$, then for $x\in\Omega$
\[
\varphi(x) = \varphi'(0)x + \sum_{k=3}^\infty \frac{\varphi^{(k)}(0)}{k!}x^k,
\]
where we apply the fact $\varphi(0)=0,\varphi''(0)=0$. Let 
\[
\bar{\varphi}(x):=\overline{\varphi'(0)}x + \sum_{k=3}^\infty \frac{\overline{\varphi^{(k)}}(0)}{k!}x^k.
\]
Then $\bar{\varphi}(\bar{x})=\overline{\varphi(x)}$ is also an analytic function of $\bar{x}$ for $\bar{x}\in\widetilde{\Omega}$ where $\widetilde{\Omega}=\{\bar{x}|x\in\Omega\}$. Considering $\gamma_\Omega$ as a function of $(z_1,z_2,\bar{z}_1, \bar{z}_2)$:
\begin{align*}
    \gamma_{\Omega}	& =\frac{1}{4\pi}\ln\left(\frac{\varphi\left(z_{1}\right)-\varphi\left(z_{2}\right)}{z_{1}-z_{2}}\right)+\frac{1}{4\pi}\ln\left(\frac{\bar{\varphi}\left(\bar{z}_{1}\right)-\bar{\varphi}\left(\bar{z}_{2}\right)}{\bar{z}_{1}-\bar{z}_{2}}\right) \\
	& \qquad-\frac{1}{4\pi}\ln\left(1-\varphi\left(z_{1}\right)\bar{\varphi}\left(\bar{z}_{2}\right)\right)-\frac{1}{4\pi}\ln\left(1-\bar{\varphi}\left(\bar{z}_{1}\right)\varphi\left(z_{2}\right)\right),
\end{align*}
then $\gamma_\Omega$ is analytic in $\Omega\times\Omega\times\widetilde{\Omega}\times\widetilde{\Omega}$. 
Changing variables as follow
\[
z_1, z_2 \to \xi, B \to \xi, b \to r_b, \theta_b, r_\xi, \theta_\xi, (\xi=r_\xi e^{i\theta_\xi},b=r_be^{i\theta_b})
\]
then expanding $\gamma_\Omega$, one can extract the non-vanishing higher order terms that do not involve angle variables such as $r_\xi^{m}r_b^{n}$. One observation is that $m=2k, n=2l$ must be even since odd order terms must live with angle variables. Then these higher order terms are actually powers of action variables with integer orders, which is exactly what we want. If we successfully find at least one of such terms, then the original Hamiltonian has the form 
\[
H=\frac{a_{1}a_{2}}{4\pi a}\ln I_\xi+\omega I_b + c I_\xi^k I_b^l + f(I) + R(\theta, I),
\]
with $|R|\le C|I|^2$ where we obtain the estimate of $R$ by applying Lemma \ref{lem:gamma}. Furthermore, let 
\[
h = \frac{a_{1}a_{2}}{4\pi a}\ln I_\xi+\omega I_b + c I_\xi^k I_b^l + f(I)
\]
be the unperturbed Hamiltonian and it will satisfy $|\det \frac{\partial^2 h}{\partial I^2}|\neq 0$ for some $I$ hence the KAM theorem \ref{thm:KAM} can apply. To do so, we can consider to expand $\widetilde{H}:=\left(H-\frac{a_{1}a_{2}}{4\pi a} \ln \frac{|\xi|^2}{2}\right)$. Recall that 
\[
z_{1}=B+\frac{a_{2}}{\sqrt{a_{1}a_{2}}}\xi,\quad z_{2}=B-\frac{a_{1}}{\sqrt{a_{1}a_{2}}}\xi,\quad B=Lb.
\]
We have 
\[
\frac{2}{a}\widetilde{H} = a_1^2\tilde{\gamma}_\Omega\left(Lb+\frac{a_{2}}{\sqrt{a_{1}a_{2}}}\xi\right)+a_2^2\tilde{\gamma}_\Omega\left(Lb-\frac{a_{1}}{\sqrt{a_{1}a_{2}}}\xi\right)+2a_1a_2\gamma\left(Lb+\frac{a_{2}}{\sqrt{a_{1}a_{2}}}\xi,Lb-\frac{a_{1}}{\sqrt{a_{1}a_{2}}}\xi\right), 
\]
where the linear operator $L$ is 
\[
L=\left(\begin{array}{cc}
\sqrt{\frac{\omega_c}{\frac{3}{2}c_{0}^{r}+c_{1}}} & -\frac{\frac{3}{2}c_{0}^{i}}{\sqrt{\omega_c (\frac{3}{2}c_{0}^{r}+c_{1}})}\\
0 & -\sqrt{\frac{\frac{3}{2}c_{0}^{r}+c_{1}}{\omega_c}}
\end{array}\right).
\]
Expressing in complex form gives 
\[
Lb = l_1 b + \bar{l}_2 \bar{b},
\]
where 
\begin{align*}
l_{1} & =-\frac{1}{2}\sqrt{\frac{\frac{3}{2}c_{0}^{r}+c_{1}}{\omega_c}}+\frac{1}{2}\left(\sqrt{\frac{\omega_c}{\frac{3}{2}c_{0}^{r}+c_{1}}}+i\frac{\frac{3}{2}c_{0}^{i}}{\sqrt{\omega_c \left(\frac{3}{2}c_{0}^{r}+c_{1}\right)}}\right),\\
l_{2} & =\frac{1}{2}\sqrt{\frac{\frac{3}{2}c_{0}^{r}+c_{1}}{\omega_c}}+\frac{1}{2}\left(\sqrt{\frac{\omega_c}{\frac{3}{2}c_{0}^{r}+c_{1}}}+i\frac{\frac{3}{2}c_{0}^{i}}{\sqrt{\omega_c \left(\frac{3}{2}c_{0}^{r}+c_{1}\right)}}\right).
\end{align*}
Theoretically we have all the ingredients to expand $\widetilde{H}$ to any fixed order in terms of $r_b, r_\xi, \theta_b, \theta_\xi$. One needs to carefully calculate the coefficient of $r_\xi^{2k}r_b^{2l}$ for some fixed $k,l$ and prove it is non-zero. Manually computing these coefficients explicitly is extremely complicated. Fortunately with the help of Mathematica, we can leave complex and repetitive tasks to computers and obtain the precise expression of the coefficients. The detail algorithm is listed as the following:
\begin{itemize}
\item Explicitly define the functions 
\begin{align*}
\gamma_{\Omega}\left(z_{1},z_{2},\bar{z}_{1},\bar{z}_{2}\right) & =\frac{1}{4\pi}\ln\left(\frac{\varphi\left(z_{1}\right)-\varphi\left(z_{2}\right)}{z_{1}-z_{2}}\right)+\frac{1}{4\pi}\ln\left(\frac{\bar{\varphi}\left(\bar{z}_{1}\right)-\bar{\varphi}\left(\bar{z}_{2}\right)}{\bar{z}_{1}-\bar{z}_{2}}\right)\\
 & \qquad-\frac{1}{4\pi}\ln\left(1-\varphi\left(z_{1}\right)\bar{\varphi}\left(\bar{z}_{2}\right)\right)-\frac{1}{4\pi}\ln\left(1-\bar{\varphi}\left(\bar{z}_{1}\right)\varphi\left(z_{2}\right)\right),\\
\tilde{\gamma}_{\Omega}\left(z,\bar{z}\right) & =\frac{1}{4\pi}\ln\left(\varphi'\left(z\right)\right)+\frac{1}{4\pi}\ln\left(\bar{\varphi}'\left(z\right)\right)-\frac{1}{2\pi}\ln\left(1-\varphi\left(z\right)\bar{\varphi}\left(\bar{z}\right)\right),\\
\tilde{H}\left(z_{1},z_{2},\bar{z}_{1},\bar{z}_{2}\right) & =\frac{1}{a}\left(\frac{a_{1}^{2}}{2}\tilde{\gamma}_{\Omega}\left(z_{1},\bar{z}_{1}\right)+\frac{a_{2}^{2}}{2}\tilde{\gamma}_{\Omega}\left(z_{2},\bar{z}_{2}\right)+a_{1}a_{2}\gamma_{\Omega}\left(z_{1},z_{2},\bar{z}_{1},\bar{z}_{2}\right)\right),
\end{align*}
where $\varphi,\bar{\varphi}$ are considered as the abstract function
object in Mathematica.
\item Evaluate $\tilde{H}$ at 
\begin{align*}
z_{1} & =l_{1}b+\bar{l}_{2}\bar{b}+\sqrt{\frac{a_{2}}{a_{1}}}\xi,\quad z_{2}=l_{1}b+\bar{l}_{2}\bar{b}-\sqrt{\frac{a_{1}}{a_{2}}}\xi,\\
\bar{z}_{1} & =\bar{l}_{1}\bar{b}+l_{2}b+\sqrt{\frac{a_{2}}{a_{1}}}\bar{\xi},\quad\bar{z}_{2}=\bar{l}_{1}\bar{b}+l_{2}b-\sqrt{\frac{a_{1}}{a_{2}}}\bar{\xi}.
\end{align*}
\item Replace $b$ by $r_{b}e^{i\theta_{b}}$, $\xi$ by $r_{\xi}e^{i\theta_{\xi}}$
and obtain the expression of $\tilde{H}$ in terms of $r_{b},r_{\xi},\theta_{b},\theta_{\xi}$.
More precisely at this stage, $\tilde{H}$ is a function of $r_{b},r_{\xi},e^{i\theta_{b}},e^{i\theta_{\xi}}$.
\item Input the desired even orders $\left(m,n\right)$ for the target term
$r_{\xi}^{m}r_{b}^{n}$. Expand $\tilde{H}$ in terms of $r_{b}$
up to $n$-th order and extract the $n$-th order coefficient:
\[
C_{n}=\text{SeriesCoefficient}\left[\text{Series}\left[\tilde{H},\left\{ r_{b},0,n\right\} \right],n\right].
\]
\item Continue expanding $C_{n}$ in terms of $r_{\xi}$ up to $m$-th order
and extract the $m$-th order coefficient:
\begin{align*}
C_{m,n} & =\text{SeriesCoefficient}\left[\text{Series}\left[C_{n},\left\{ r_{\xi},0,m\right\} \right],m\right]\\
 & \qquad/.\left\{ \varphi\left(0\right)\to0,\bar{\varphi}\left(0\right)\to0,\varphi''\left(0\right)\to0,\bar{\varphi}''\left(0\right)\to0\right\} .
\end{align*}
We remark here $C_{m,n}$ consists of some constant and power series
of $e^{i\theta_{b}},e^{i\theta_{\xi}}$. The last part in brackets
means plugging these special values. 
\item Integrate $C_{m,n}$ in angle variable to eliminate the angle part
and finally obtain the coefficient of $r_{\xi}^{m}r_{b}^{n}$:
\[
C_{m,n}=\text{Integrate}\left[\text{Simplify}\left[C_{m,n}\right],\left\{ \theta_{b},0,2\pi\right\} ,\left\{ \theta_{\xi},0,2\pi\right\} \right]/4/\pi^{2}.
\]
\end{itemize}
According to \eqref{eq:H_expand_2nd}, we can easily imply that the coefficient of $r_b^2$ is $\frac{\omega}{2}$ which can be used to verify the algorithm above for $m=0, n=2$. On the other hand, we find both $\omega$ and $C_{4,2}$ the coefficient of $r_\xi^4r_b^2$ depending only on $\varphi'(0), \varphi'''(0)$. Recall that 
\[
\omega_c = \frac{\omega}{a} = \sqrt{c_1^2-\frac{9}{4}|c_0|^2}=\frac{\sqrt{4 |\varphi'(0)|^4 -\left|\frac{\varphi'''(0)}{\varphi'(0)}\right|^2}}{4 \pi }.
\]
We obtain $C_{4,2}$ via Mathematica and replacing $\sqrt{4 |\varphi'(0)|^4 -\left|\frac{\varphi'''(0)}{\varphi'(0)}\right|^2}$ by $4\pi\omega_c$ yields a readable clean expression
\begin{align*}
C_{4,2} & =\frac{a\left|\varphi'\left(0\right)\right|^{2}}{32\pi^{2}\omega_{c}}\left(12\left|\varphi'\left(0\right)\right|^{6}+\left|\varphi'''\left(0\right)\right|^{2}\right.\\
 & \quad\left.+2i\Im\left(\frac{\varphi'''\left(0\right)}{\varphi'\left(0\right)}\right)\left(\bar{\varphi}'\left(0\right)\varphi'''\left(0\right)-\varphi'\left(0\right)\bar{\varphi}'''\left(0\right)\right)-2\Re\left(\frac{\varphi'''\left(0\right)}{\varphi'\left(0\right)}\right)\left(\bar{\varphi}'\left(0\right)\varphi'''\left(0\right)+\varphi'\left(0\right)\bar{\varphi}'''\left(0\right)\right)\right)
\end{align*}
The expression can be further simplified by noticing 
\begin{align*}
 & 2i\Im\left(\frac{\varphi'''\left(0\right)}{\varphi'\left(0\right)}\right)\left(\bar{\varphi}'\left(0\right)\varphi'''\left(0\right)-\varphi'\left(0\right)\bar{\varphi}'''\left(0\right)\right)-2\Re\left(\frac{\varphi'''\left(0\right)}{\varphi'\left(0\right)}\right)\left(\bar{\varphi}'\left(0\right)\varphi'''\left(0\right)+\varphi'\left(0\right)\bar{\varphi}'''\left(0\right)\right)\\
= & \left(\frac{\varphi'''\left(0\right)}{\varphi'\left(0\right)}-\frac{\bar{\varphi}'''\left(0\right)}{\bar{\varphi}'\left(0\right)}\right)\left(\bar{\varphi}'\left(0\right)\varphi'''\left(0\right)-\varphi'\left(0\right)\bar{\varphi}'''\left(0\right)\right)-\left(\frac{\varphi'''\left(0\right)}{\varphi'\left(0\right)}+\frac{\bar{\varphi}'''\left(0\right)}{\bar{\varphi}'\left(0\right)}\right)\left(\bar{\varphi}'\left(0\right)\varphi'''\left(0\right)+\varphi'\left(0\right)\bar{\varphi}'''\left(0\right)\right)\\
= & -4\left|\varphi'''\left(0\right)\right|^{2}.
\end{align*}
Finally, we arrive at 
\begin{align*}
C_{4,2} & =\frac{a\left|\varphi'\left(0\right)\right|^{2}}{32\pi^{2}\omega_{c}}\left(12\left|\varphi'\left(0\right)\right|^{6}-3\left|\varphi'''\left(0\right)\right|^{2}\right)\\
 & =\frac{3a\left|\varphi'\left(0\right)\right|^{4}}{32\pi^{2}\omega_{c}}\left(4\left|\varphi'\left(0\right)\right|^{4}-\left|\frac{\varphi'''\left(0\right)}{\varphi'\left(0\right)}\right|^{2}\right)\\
 & =\frac{3}{2}a\left|\varphi'\left(0\right)\right|^{4}\omega_{c}
\end{align*}
which is never zero. Then the Hamiltonian can be written as 
\[
H=\underbrace{\frac{a_{1}a_{2}}{4\pi a}\ln I_\xi+\omega I_b + 8C_{4,2} I_\xi^2 I_b + f(I)}_{h(I)} + R(\theta, I),
\]
with $|R|\le C|I|^2$ and $f(I)$ is a power series of $I$. One can verify that for $I^*=(I_\xi^*, I_b^*), I_\xi^*\neq0$, it holds that $|\det \frac{\partial^2 h}{\partial I^2}(I^*)|\neq 0$. Because the set of frequencies that satisfy the Diophantine condition is dense, we can choose $I^*$ small and meanwhile in a dense set, and it still meets the condition of KAM theorem. More precisely, there exists a small enough $\epsilon_0$ such that for all $0<\epsilon\le\epsilon_0$, with the possible exception of a set of Lebesgue measure zero, for almost all $|I(0)|=|I^*|\le \epsilon$, it holds that $|I(t)|\le \mu \epsilon$ for all $t>0$. The choice of $I^*$ implies the condition of $z_1(0), z_2(0)$ and proves this theorem.

\subsection{Proof of Theorem \ref{main result 2}}
The proof for the critical case is totally different but much easier. We will need the following bootstrap Lemma. See for example \cite{bahouri1999high} 
or \cite{ibrahim2003solutions}.
\begin{lem}[bootstrap]
Let $0<T<\infty$, $c_1,c_2>0$, $\nu>1$ and 
$$\Theta:[0,T)\to[0,\infty)$$
be a continuous function satisfying
$$
\Theta(t)\leq c_1+c_2\Theta(t)^\nu$$
with
$$
c_1<(1-\frac1\alpha)\frac1{(\nu c_2)^\frac{1}{\nu-1}},\quad\mbox{and}\quad \Theta(0)<\frac1{(\nu c_2)^\frac1{\nu-1}}$$
then, for any $t\in[0,T)$, we have 
$$
\Theta(t)\leq\frac\nu{\nu-1}c_1.
$$
\end{lem}

Recall the conserved energy of the system is
\[
H=\frac{1}{2}\sum_{i=1}^{2}a_{i}^{2}\gamma_\Omega\left(z_{i},z_{i}\right)+a_{1}a_{2}\left(\gamma_\Omega\left(z_{1},z_{2}\right)+\frac{1}{4\pi}\ln\left(\left|z_{1}-z_{2}\right|^{2}\right)\right).
\] Applying smooth mapping $\tilde{H}=e^{\frac{4\pi}{a_{1}a_{2}}H}$
and Lemma $\text{\ref{lem:gamma}}$ yield
\begin{align}
    \tilde{H} & =\left|z_{1}-z_{2}\right|^{2}e^{2\pi\left(2\gamma_\Omega\left(z_{1},z_{2}\right)+\frac{a_{1}}{a_{2}}\gamma_\Omega\left(z_{1},z_{1}\right)+\frac{a_{2}}{a_{1}}\gamma_\Omega\left(z_{2},z_{2}\right)\right)} \nonumber \\
    &=\left|\xi\right|^{2}e^{K_{0}+P\left(z_1,z_2,\bar{z}_1,\bar{z}_2\right)+K_{1}\left(E_{\gamma_\Omega}\left(z_1,z_2\right)+E_{\gamma_\Omega}\left(z_1,z_1\right)+E_{\gamma_\Omega}\left(z_2,z_2\right)\right)} \label{tilde_quantity}
\end{align}
where $K_0,K_1$ are two positive constants only depending on $\Omega$ and $P\left(z_1,z_2,\bar{z}_1,\bar{z}_2\right)$ is a quadratic form. At $t=0$, $P$ is controlled by 
\[
|P\left(z_1,z_2,\bar{z}_1,\bar{z}_2\right)|\le C\max\left\{|z_1(0)|,|z_2(0)|\right\}^2\le C\varepsilon^2,
\]
and the error terms are all bounded by
\[
\left|E_{\gamma_\Omega}\left(t=0\right)\right|\le C\max\left\{|z_1(0)|,|z_2(0)|\right\}^{3} \le C\varepsilon^{3}.
\]
So initially, we have
\[
\tilde{H}\left(t=0\right)\le 4\varepsilon^{2}e^{K_{0}+C\varepsilon^{2}}.
\]
For $t<t_{\varepsilon,\beta}$ it holds that $|z_k(t)|\le \varepsilon^\beta$ so that 
\[
\tilde{H}(t) \ge |\xi(t)|^2e^{K_0-C\varepsilon^{2\beta}}.
\]
The consevsation of $\tilde{H}$ gives 
\[
|\xi(t)|^2e^{K_0-C\varepsilon^{2\beta}} \le 4\varepsilon^2e^{K_0+C\varepsilon^2},
\]
which implies
\begin{equation}
    \label{eq:tH_estimate_strip}
    |\xi(t)|\le 2\varepsilon e^{\frac{C}{2}(\varepsilon^2+\varepsilon^{2\beta})}\le 2\varepsilon e^{C\varepsilon^{2\beta}}.
\end{equation}
Recall the full equation for $B$ is of the following form
\begin{align*}
\frac{d}{dt}B= & i\frac{a_{1}^{2}}{2\left(a_{1}+a_{2}\right)}\nabla\tilde{\gamma}_\Omega\left(z_{1}\right)+i\frac{a_{2}^{2}}{2\left(a_{1}+a_{2}\right)}\nabla\tilde{\gamma}_\Omega\left(z_{2}\right)\\
 & +\frac{ia_{1}a_{2}}{\left(a_{1}+a_{2}\right)}\left(\nabla_{1}\gamma_\Omega\left(z_{1},z_{2}\right)+\nabla_{1}\gamma_\Omega\left(z_{2},z_{1}\right)\right).
\end{align*}
Applying Lemma $\text{\ref{lem:grad_gamma}}$ yields
\begin{equation}
\frac{d}{dt}B=i\left(a_{1}+a_{2}\right)\left(\frac{3}{2}\bar{c}_{0}\bar{B}+c_{1}B\right)+\mathcal{O}\left(\left|z_{1}\right|^{2}+\left|z_{2}\right|^{2}\right).\label{eq:B}
\end{equation}
Considering $B$ as a vector and letting $c_{0}^{R}$ and $c_{0}^{I}$
be the real and imaginary part of $c_{0}$, the solution to $\left(\text{\ref{eq:B}}\right)$
is
\begin{equation}
B\left(t\right)=e^{\mathcal{R}t} B\left(0\right)+\int_{0}^{t}e^{\left(t-s\right)\mathcal{R}}\mathcal{O}\left(\left|z_{1}\left(s\right)\right|^{2}+\left|z_{2}\left(s\right)\right|^{2}\right)ds,\label{eq:reason_B_rotation}
\end{equation}
where $\mathcal{R}$ is the following matrix:
\[
\mathcal{R}=\left(\begin{array}{cc}
\frac{3}{2}\left(a_{1}+a_{2}\right)c_{0}^{I} & \left(a_{1}+a_{2}\right)\left(\frac{3}{2}c_{0}^{R}-c_{1}\right)\\
\left(a_{1}+a_{2}\right)\left(\frac{3}{2}c_{0}^{R}+c_{1}\right) & -\frac{3}{2}\left(a_{1}+a_{2}\right)c_{0}^{I}
\end{array}\right),
\]
whose eigenvalues are $\pm\frac{a_{1}+a_{2}}{2}\sqrt{9\left|c_{0}\right|^{2}-4c_{1}^{2}}$. In critical case, $2c_{1}=3\left|c_{0}\right|$ which is exactly the condition $2\left|\varphi'(0)\right|^{3}=\left|\varphi'''(0)\right|$.
Therefore the linear operator $e^{\mathcal{R}t}$ reads
\[
e^{\mathcal{R}t} = I+\mathcal{R}t=\left(\begin{array}{cc}
1 +\frac{3(a_1+a_2)}{2}c_0^It& (a_1+a_2)(\frac{3}{2}c_0^R-c_1)t\\
(a_1+a_2)(\frac{3}{2}c_0^R+c_1)t & 1 -\frac{3(a_1+a_2)}{2}c_0^It
\end{array}\right).
\]
Plugging it into \eqref{eq:reason_B_rotation} and applying the estimate of $\left|\xi\right|$ \eqref{eq:tH_estimate_strip} yield
\begin{align*}
    \left|B\left(t\right)\right|	& \le C\left|B\left(0\right)\right|\left(1+t\right)+C\sup_{0\le s\le t}\left(\left|z_{1}\left(s\right)\right|^{2}+\left|z_{2}\left(s\right)\right|^{2}\right)\int_{0}^{t}1+\left(t-s\right)ds\\
	&\le C\left|B\left(0\right)\right|\left(1+t\right)+C\left(t+t^{2}\right)\sup_{0\le s\le t}\left(\left|z_{1}\left(s\right)\right|^{2}+\left|z_{2}\left(s\right)\right|^{2}\right) \\
    & \le Ct\varepsilon + Ct^2 \sup_{0\le s\le t}\left(\left|B\left(s\right)\right|^{2}+\left|\xi\left(s\right)\right|^{2}\right) \\
    & \le Ct\varepsilon+4C\varepsilon^{2}t^{2}e^{2C\varepsilon^{2\beta}}+Ct^2\sup_{0\le s\le t}\left|B\left(s\right)\right|^{2}.
\end{align*}
Let $f(t)=\sup_{0\le s\le t}\left|B\left(s\right)\right|$. Taking the supremum in time in the above inequality gives
\begin{align*}
    f\left(t\right)	& \le Ct\varepsilon+4C\varepsilon^{2}t^{2}e^{2C\varepsilon^{2\beta}}+Ct^2f^2\left(t\right)\\
    & \le C\left(t+t^{2}e^{C\varepsilon^{2\beta}}\right)\varepsilon + Ct^2f^2\left(t\right) \\
    & \le 2CT^2\varepsilon + CT^2 f^2\left(t\right),
\end{align*}
where $T>t$ to be fixed. If $T=\varepsilon^{-\alpha}$ with $\alpha<\frac{1}{4}$ and the conditions of the {\it bootstrap} Lemma (taking $\nu=2$) are satisfied
\begin{equation}
     f(0)=|B(0)|\le \varepsilon < \frac{1}{2CT^{2}}, \label{eq:bootstrap_cond_1}
\end{equation}
\begin{equation}
     2CT^{2}\varepsilon<\frac{1}{4CT^{2}}, \label{eq:bootstrap_cond_2}
\end{equation}
then we get for any $0\leq t\leq \varepsilon^{-\alpha}$, 
\begin{equation}
\label{eq:B_estimate_strip}
|B(t)|\le f(t)\leq 2CT^{2}\varepsilon\leq\frac{1}{2}\varepsilon^{1-2\alpha}.
\end{equation}
Furthermore, if in addition $\alpha<\frac{1-\beta}{2}$ then \eqref{eq:B_estimate_strip} implies $|B(t)| < \frac{1}{2}\varepsilon^{\beta}$.
In summary, let $\alpha=\min\left(\frac{1-\beta}{2},\frac{1}{4}\right)$ and if bootstrap condition \eqref{eq:bootstrap_cond_1}, \eqref{eq:bootstrap_cond_2} are satisfied, then we obtain that for $0\leq t\leq \varepsilon^{-\alpha}$, it holds that 
\[
|B(t)| < \frac{1}{2}\varepsilon^{\beta}.
\]
Putting \eqref{eq:tH_estimate_strip} and \eqref{eq:B_estimate_strip} together yields the estimate of $|z_k(t)|$:
\[
|z_k(t)|\le |B(t)|+\frac{\max\left(\left|a_{1}\right|,\left|a_{2}\right|\right)}{\left|a_{1}+a_{2}\right|}|\xi(t)|\le\frac{1}{2}\varepsilon^\beta + C\varepsilon e^{C\varepsilon^{2\beta}}.
\]
We expect the above upper bound is less than $\varepsilon^\beta$, the bootstrap argument can apply and we arrive at the conclusion $t_{\varepsilon, \beta}>C\varepsilon^{-\alpha}$. At the very last, it is necessary to find the small $\varepsilon$ such that bootstrap condition \eqref{eq:bootstrap_cond_1}, \eqref{eq:bootstrap_cond_2} and 
\[
C\varepsilon e^{C\varepsilon^{2\beta}} <\frac{1}{2}\varepsilon^\beta
\]
are satisfied. These conditions are equivalent to 
\begin{align*}
    \varepsilon^{1-2\alpha} & < \frac{1}{2C}, \\
    \varepsilon^{1-4\alpha} & < \frac{1}{8C^2}, \\
    2C\varepsilon^{1-\beta}e^{C\varepsilon^{2\beta}} & < 1.
\end{align*}
The first two are obvious for small $\varepsilon$ and the last one is true as long as there exists $0<\varepsilon\le\varepsilon_0$ such that
\[
2C\varepsilon_0^{1-\beta}e^{C\varepsilon_0^{2\beta}}<1
\]
holds. Obviously, such small $\varepsilon_0$ exists. 
\section{Examples}
We give some examples of domains with stable, unstable stationary points. Recall that the stable stationary point should satisfy 
\[
|\varphi'''(0)| < 2|\varphi'(0)|^3.
\]
The critical case will satisfy the equality. 
\begin{figure}[ht]
\centering
\includegraphics[scale=0.5]{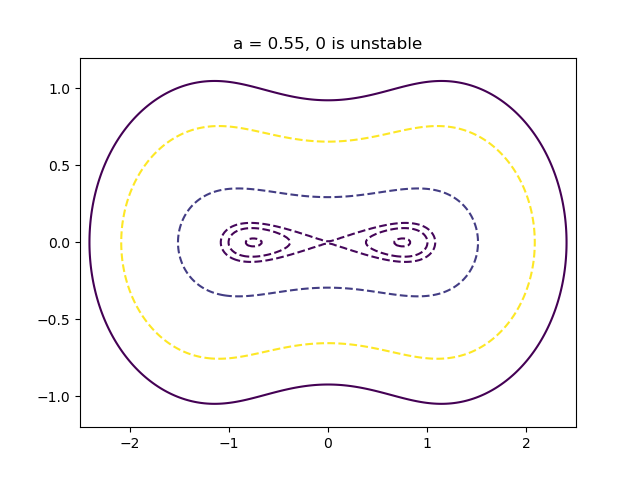}\includegraphics[scale=0.5]{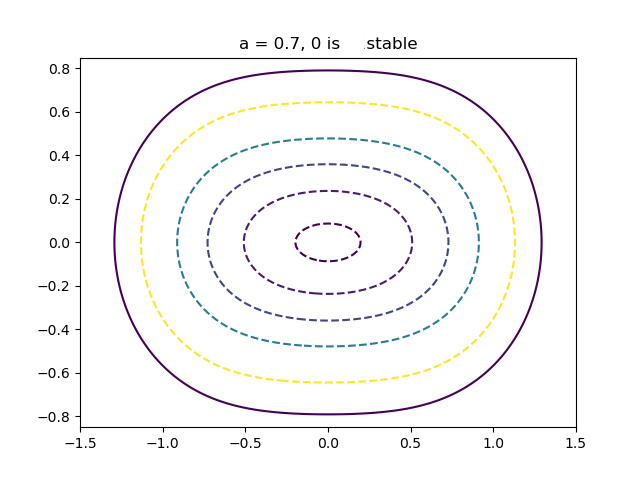}
\caption{$\varphi(z) = a\left(\tan\left(iz\right)+\tan\left(i\frac{z}{2}\right)\right)$, $0$ is the unstable stationary point for $\frac{1}{2}<a<\frac{1}{\sqrt{3}}$ and stable stationary point for $a>\frac{1}{\sqrt{3}}$. \label{fig:bounded_unstable_stable}}
\end{figure}
\begin{itemize}
    \item Critical case, the infinite strip between $Im(z)=\pm 1$. The conformal map 
\[
\varphi(z) = \tan\left(i\frac{\pi z}{4}\right)
\]
maps the strip into $\mathbb{D}$. It is easy to verify $|\varphi'''(0)|=2|\varphi'(0)|^3$. Actually, the whole real line is stationary and all satisfy the critical stable condition.
\item Bounded domains with stable or unstable stationary points. The examples are inspired by the previous critical case. Let the conformal map be 
\[
\varphi(z)=a\left(\tan\left(iz\right)+\tan\left(i\frac{z}{2}\right)\right),
\]
where $a>\frac{1}{2}$ so that $\varphi$ maps a bounded domain $\Omega\to\mathbb{D}$. It is easy to check $\varphi(0)=0,\varphi''(0)=0$ so $0$ is the stationary point. And it holds that 
\[
|\varphi'''(0)|-2|\varphi'(0)|^3=\frac{9a(1-3a^2)}{4}.
\]
Therefore for $\frac{1}{2}<a<\frac{1}{\sqrt{3}}$, $0$ is unstable while for $a>\frac{1}{\sqrt{3}}$, $0$ is stable. In addition, when $\frac{1}{2}<a<\frac{1}{\sqrt{3}}$ the domain has another two stable stationary points on the real axis (see Figure \ref{fig:bounded_unstable_stable}). 
\item Martin Donati \cite{Donati2024} found a family of biconvex hexagonal domains with $0$ the stationary point. The domains are given by $\varphi^{-1}:\mathbb{D}\to \Omega_\delta$ such that $\varphi^{-1}(0)=0$ and
\[
\left(\varphi^{-1}\right)'(z)=\frac{(z^2-1)^{3\delta-1}}{(z^6-1)^\delta}.
\]
It is shown that $0$ is stable for $\frac{1}{2}\le \delta< \frac{2}{3}$, critical for $\delta=\frac{2}{3}$ and unstable for $\frac{2}{3}<\delta< 1$.
\end{itemize}
 
\section{Acknowledgement}

Shengyi Shen was supported by NSFC Grant(12201324). Ruixun Qin was
partially supported by Ningbo University scholarship.

\appendix

\section{Appendix}

For the sake of completeness, we list some results on the power series expansion
of holomorphic functions and the control of their remainders (see \cite{hormander1973introduction}) as well as the KAM theorem.
Let $\Omega_{j}\subset\mathbb R^2\sim\mathbb{C},j=1,\cdots$ be $n$ bounded domains in $\mathbb R^2\sim\mathbb C$, each with $C^{1}$ Jordan closed curve boundary. Let $\tilde{\Omega}:=\Omega_{1}\times \Omega_{2}\times\cdots\times \Omega_{n}$
and $f:\tilde{\Omega}\to\mathbb{C}$ be holomorphic and denote $\mathbf{z}=\left(z_{1},\cdots,z_{n}\right)\in\mathbb{C}^{n}$
the complex vector. By the Cauchy-Riemann equation, $f$ should be independent
from $\bar{z}_{j}$. By repeating Cauchy's integral formula of one
variable repeatedly one has the representation formula for any derivatives in terms of contour integrals of $f$: for $\mathbf{z}\in\tilde{\Omega}$
\begin{equation}
\partial^{\alpha}f\left(\mathbf{z}\right)=\frac{\alpha!}{\left(2\pi i\right)^{n}}\int_{\partial \Omega_{1}}\cdots\int_{\partial \Omega_{n}}\frac{f\left(\mathbf{w}\right)}{\left(\mathbf{w}-\mathbf{z}\right)^{\alpha+\left(1,\cdots,1\right)}}d\mathbf{w},\label{eq:cauchy_integral_formula}
\end{equation}
where $\alpha$ is the multi-index. Furthermore if $\mathbf{z}$ is in a polydisc of $\tilde{\Omega}$, that is
\[
\mathbf{z}\in\prod_{i=1}^nD\left(z_j,r_j\right)\subset\tilde{\Omega},
\]
then \eqref{eq:cauchy_integral_formula} becomes
\[
\partial^{\alpha}f\left(\mathbf{z}\right)=\frac{\alpha!}{\left(2\pi i\right)^{n}}\int_{\partial D\left(z_1,r_1\right)}\cdots\int_{\partial{D\left(z_n,r_n\right)}}\frac{f\left(\mathbf{w}\right)}{\left(\mathbf{w}-\mathbf{z}\right)^{\alpha+\left(1,\cdots,1\right)}}d\mathbf{w}.
\]
Taking $M=\max_{\mathbf{w}\in{\prod_{i=1}^n\partial D\left(z_j,r_j\right)}}\left|f\left(\mathbf{w}\right)\right|$ one obtains the multi-variable version of Cauchy's estimate \cite{ebeling2007functions}: 
\begin{equation}
\left|\partial^{\alpha}f\left(\mathbf{z}\right)\right|\le\frac{M\alpha!}{\prod_{i=1}^{n}r_{i}^{\alpha_{i}}}.\label{eq:cauchy_estimate}
\end{equation}
From Cauchy's integral formula \eqref{eq:cauchy_integral_formula},
one can derive that $f$ can be uniquely expanded to the power series.
In addition, the remainder estimate is done with the help of \eqref{eq:cauchy_estimate}. 
\begin{prop}
\label{prop:cauchy=000020estimate} Let $\mathbf{c}\in\tilde{\Omega}$
be an interior point such that each disk centered at $c_{j}$ with radius $r_{j}$ satisfies $D\left(c_{j},r_{j}\right)\subset \Omega_{j}$. A holomorphic
function $f$ on $\tilde{\Omega}$ can be uniquely expanded at point
$\mathbf{c}$:
\[
f\left(\mathbf{z}\right)=\sum_{\left|\alpha\right|=0}^{\infty}\frac{\partial^{\alpha}f\left(\mathbf{c}\right)}{\alpha!}\left(\mathbf{z}-\mathbf{c}\right)^{\alpha}.
\]
Furthermore, if $\mathbf{z}\in\prod_{i=1}^nD\left(z_j,r_j\right)$ one can truncate the power series at order $k$:
\[
f\left(\mathbf{z}\right)=\sum_{\left|\alpha\right|=0}^{k}\frac{\partial^{\alpha}f\left(\mathbf{c}\right)}{\alpha!}\left(\mathbf{z}-\mathbf{c}\right)^{\alpha}+R_{k}\left(\mathbf{z}\right),
\]
the remainder is controlled by 
\[
\left|R_{k}\left(\mathbf{z}\right)\right|\le M\sum_{\left|\alpha\right|\ge k+1}\prod_{i=1}^{n}\left|\frac{z_{i}-c_{i}}{r_{i}}\right|^{\alpha_{i}},
\]
where 
\[
M=\max_{\mathbf{w}\in\prod_{i=1}^nD\left(z_j,r_j\right)}\left|f\left(\mathbf{w}\right)\right|.
\]
\end{prop}
Before we state the classic KAM theorem, the norm that is used to measure the perturbation is 
\[
\norm{f}_{A_{\rho,\sigma}(I^*)} := \sup_{(I,\theta)\in A_{\rho,\sigma}(I^*)} |f(I,\theta)|,
\]
where
\[
A_{\rho,\sigma}(I^*) := \left\{(I,\theta)\big| |I-I^*|<\rho,\Im(\theta) < \sigma \right\}.
\]
The classic KAM theorem is the following (see ref \cite{Wayne1994}, \cite{de2003tutorial}, \cite{Kolmogorov1954})
\begin{thm}
\label{thm:KAM}Let the Hamiltonian $H\left(I,\theta\right):\mathbb{R}^{n}\times\mathbb{T}^{n}\to\mathbb{R}$
be an analytic function such that 
\[
H\left(I,\theta\right)=h\left(I\right)+R\left(I,\theta\right),
\]
and $\omega^{*}=\nabla h\left(I^{*}\right)\in\mathbb{R}^{n}$ for
some $I^{*}$ satisfies the Diophantine condition: for all $k\in \mathbb{Z}^n-\{0\}$ there exists $\nu>0$ such that
\[
|\omega^*\cdot k|\ge \frac{C}{|k|^\nu}.
\]
Assuming that for $I$ in a neighborhood of $I^{*}$
we have 
\[
\left|\det\frac{\partial^{2}h}{\partial I^{2}}\left(I\right)\right|\ge\kappa>0,
\]
then if $\norm R_{A_{\rho,\sigma}\left(I^{*}\right)}$ is small enough,
the Hamiltonian system for $H\left(I,\theta\right)$ admits a quasi-periodic
solution with frequency $\omega^{*}$ and is arbitrary close to the
torus $\left\{ I^{*}\right\} \times\mathbb{T}^{n}$.
\end{thm}

\bibliographystyle{plain}
\bibliography{two_point_votex}

\end{document}